\documentclass[11pt]{article}

\usepackage{iftex}
\ifpdf\else
  \ifXeTeX\else
    \PassOptionsToPackage{dvipdfmx}{geometry}
    \PassOptionsToPackage{dvipdfmx}{graphicx}
    \PassOptionsToPackage{dvipdfmx}{color}
    \PassOptionsToPackage{dvipdfmx}{hyperref}
  \fi
\fi

\usepackage[a4paper,margin=32mm]{geometry}
\usepackage{amsmath,amssymb,amsthm,mathtools,mathrsfs}
\usepackage{pxfonts}
\usepackage{microtype}
\usepackage[colorlinks=true,citecolor=blue,linkcolor=blue,urlcolor=blue]{hyperref}

\newtheorem{theorem}{Theorem}[section]
\newtheorem{proposition}[theorem]{Proposition}

\newtheorem{corollary}[theorem]{Corollary}
\newtheorem{assumption}[theorem]{Assumption}
\theoremstyle{remark}

\newtheorem{example}[theorem]{Example}

\newcommand{\E}{\mathsf E}
\newcommand{\Q}{\mathsf Q}
\newcommand{\Pbs}{P_{\mathrm{BS}}}
\newcommand{\cE}{\mathcal E}

\title{Yet another asymptotic formula for implied volatility}
\author{Masaaki Fukasawa\thanks{The University of Osaka.
Email: \href{mailto:fukasawa@sigmath.es.osaka-u.ac.jp}{\nolinkurl{fukasawa@sigmath.es.osaka-u.ac.jp}}}}
\date{}
\hypersetup{
  pdftitle={Yet another asymptotic formula for implied volatility},
  pdfauthor={Masaaki Fukasawa}
}

\begin{document}

\maketitle

\begin{abstract}
We derive a first-order representation of Black--Scholes implied
variance in a continuous local martingale model.  Total implied variance is
the conditional expectation of the quadratic variation of the log
price given its terminal value, up to a smaller-order term, for
bounded standardized log-strikes.  The framework incorporates small volatility-of-volatility, fast mean-reverting, and short-maturity asymptotics.
\end{abstract}

\section{Introduction}

The implied volatility has been a basic quantity in financial engineering, which still is not yet fully understood. Let $S$ be a
strictly positive continuous local martingale modeling an underlying asset price and let
\[
 A_T=\langle\log S\rangle_T.
\]
The quantity $A_T$ is the realized total variance of the log price.
The purpose of this paper is to verify a simple approximation
\begin{equation}\label{eq:intro-formula}
 \widehat w(k,T)
 \approx \E[A_T\mid S_T=S_0e^k]
\end{equation}
for the Black--Scholes total implied variance $\widehat w(k,T)$
defined from put prices,
or equivalently,
\[
 \widehat\sigma(k,T)^2
 \approx \frac1T\E[A_T\mid S_T=S_0e^k]
\]
for the Black--Scholes annualized implied variance $\widehat \sigma(k,T)^2$,
when $A_T/T$ is close to a deterministic value, under quantitative
concentration and Gaussian smoothing conditions, and for bounded
standardized log-strikes.
Here, we work under a pricing measure $\Q$, with zero interest and dividend
rates, and write $\E$ for expectation under $\Q$.  

Formula \eqref{eq:intro-formula} was conjectured in
Fukasawa~\cite[Remark~2.5]{Fukasawa2026}, where it was obtained by a
formal substitution in a martingale expansion for stochastic
volatility.  We give a rigorous first-order
validation extending the
conditional Gaussian smoothing and stopped Black--Scholes argument
of \cite[Section~4.1]{Fukasawa2026}, retaining the terminal conditional
expectation before taking limits.  The present formulation allows
varying maturities and continuous quadratic-variation clocks without
an absolute-continuity assumption.

The conditional expectation in \eqref{eq:intro-formula} averages realized total variance over paths ending at the strike. A related interpretation arises from Gamma-weighted representations of implied variance given by
Gatheral~\cite{Gatheral} and proved by  Keller--Ressel and
Teichmann~\cite{KT}.  A related exact
identity is used in Bergomi~\cite[Section~2.4.1]{Bergomi} to interpret implied
variance through the profit and loss of a delta-hedged option.

\section{Framework and main results}
\label{sec:results}

\subsection{Black--Scholes notation}

Let $\Phi$ and $\phi$ denote the standard normal distribution function
and density.  For $s,K>0$ and total variance $q>0$, let
\begin{equation}\label{eq:BS-put}
 \Pbs(s,K,q)=K\Phi(-d_2)-s\Phi(-d_1),
 \qquad
 d_{1,2}=\frac{\log(s/K)}{\sqrt q}\pm\frac{\sqrt q}{2}.
\end{equation}
At $q=0$ we set $\Pbs(s,K,0)=(K-s)_+$.  Regarding $q$ as a time
variable, the Black--Scholes equation and the Gamma--Vega identity are
\begin{equation}\label{eq:BS-PDE}
 \partial_q\Pbs(s,K,q)
 =\frac12s^2\partial_{ss}\Pbs(s,K,q)
 =\frac{K}{2\sqrt q}\phi(d_2).
\end{equation}
We use
\begin{equation}\label{eq:BS-Gamma}
 \Gamma_{\mathrm{BS}}(s,K,q)
 :=\partial_{ss}\Pbs(s,K,q),
 \qquad q>0,
\end{equation}
for Black--Scholes Gamma parameterized by remaining total variance.

For a model price
\[
 P(K,T)=\E[(K-S_T)_+],
\]
the total implied variance $\widehat w(k,T)$, where
$K=S_0e^k$, is defined by
\begin{equation}\label{eq:implied-total-variance}
 P(K,T)=\Pbs(S_0,K,\widehat w(k,T)).
\end{equation}
Thus $\widehat w(k,T)=T\widehat\sigma(k,T)^2$.
Throughout, implied variance is defined from put prices.

\subsection{Framework}
\label{subsec:framework}

We now specify the asymptotic framework.  All limits and asymptotic statements below are understood as $n\to\infty$,
unless otherwise specified.
For every integer $n\ge1$, let $T_n>0$ and let $S^n$ be a
strictly positive continuous $\Q$-local martingale on $[0,T_n]$, with
$S_0^n=S_0$.  Write
\begin{equation}\label{eq:stochastic-logarithm}
 \frac{S^n}{S_0}=\cE(M^n),
 \qquad
 M_t^n=\int_0^t\frac{\mathrm dS_u^n}{S_u^n},
\end{equation}
where $M^n$ is a continuous local martingale.  Put
\begin{equation}\label{eq:qv}
 A_t^n=\langle M^n\rangle_t,
 \qquad A^n=A_{T_n}^n.
\end{equation}
Then
\[
 \log\frac{S_{T_n}^n}{S_0}
 =M_{T_n}^n-\frac12A^n,
 \qquad
 A^n=\langle\log S^n\rangle_{T_n}.
\]
We assume throughout that $A^n$ is strictly positive almost
surely and integrable, and set
\begin{equation}\label{eq:mean-total-variance}
 Z^n:=\frac{A^n}{\E[A^n]},
 \qquad
 v_n:=\E[A^n]\in(0,\infty).
\end{equation}

\begin{assumption}[Quadratic-variation perturbation]\label{ass:qv}
There exists a deterministic sequence $r_n>0$ such that
\begin{equation}\label{eq:scales}
 r_n(1+\sqrt{v_n})\longrightarrow0
\end{equation} 
and that
\begin{equation}\label{eq:short-time-localization}
 \Q\bigl(|Z^n-1|>\eta\bigr)
 =o\!\left(r_n\sqrt{v_n}\right)
\end{equation}
for every $\eta>0$,  and, with
\begin{equation}\label{eq:Y-definition}
 Y^n
 :=\frac{Z^n-1}{r_n\sqrt{Z^n}},
\end{equation}
the family
$ \left\{Y^n\right\}_{n\ge1}$
is uniformly integrable. 
\end{assumption}

Since $\E[Z^n]=1$,  the family
$\{\sqrt{Z^n}\}$ is bounded in $L^2$.  Moreover,
\[
 (Z^n)^{-1/2}
 =\sqrt{Z^n}-r_n Y^n.
\]
Since $r_n\to0$, the sequence $\{r_n\}$ is bounded.  Thus
the uniform integrability of $\{Y^n\}$ and the preceding identity also make
$\{(Z^n)^{-1/2}\}$ uniformly integrable.

When $v_n$ is bounded away from zero, the condition
\eqref{eq:short-time-localization} is automatic from
the uniform integrability of $\{Y^n\}$.
   To see this, on $\{|Z^n-1|>\eta\}$,
\[
 |Y^n|\ge\frac{c_\eta}{r_n},\qquad
 c_\eta := \inf_{\substack{z>0\\|z-1|>\eta}}
                  \frac{|z-1|}{\sqrt{z}} > 0.
\]
The uniform integrability gives a probability of $o(r_n)$, which is
$o(r_n\sqrt{v_n})$ in that case.  When
$v_n\to0$,
\eqref{eq:short-time-localization} is the additional control required
at the vanishing option-price scale.

The quantitative tail condition has convenient sufficient forms:

\begin{proposition}\label{prop:Lp}
Suppose that \eqref{eq:scales} holds and that for some $p>1$,
\begin{equation}\label{eq:Lp-condition}
 \sup_n\E[|Y^n|^p]<\infty
\end{equation}
and
\begin{equation}\label{eq:Lp-rate}
 \frac{r_n^{p-1}}{\sqrt{v_n}}
 \longrightarrow0.
\end{equation}
Then Assumption~\ref{ass:qv} is met.
\end{proposition}

\begin{proof}
The $L^p$ bound implies the  uniform integrability of $\{Y^n\}$.  Moreover, on
$\{|Z^n-1|>\eta\}$,
\[
 |Y^n|
 =\frac{|Z^n-1|}
        {r_n\sqrt{Z^n}}
 \ge\frac{c_\eta}{r_n}.
\]
Markov's inequality gives a probability
$O(r_n^p)$.  Equation \eqref{eq:Lp-rate} completes the proof.
\end{proof}

We introduce a structural condition:

\begin{assumption}[Independent Brownian smoothing component]
\label{ass:smoothing}
There are continuous local martingales $J^n$ and $N^n$,
starting from zero, such that
\begin{equation}\label{eq:orthogonal-decomposition}
 M^n=J^n+N^n,
 \qquad \langle J^n,N^n\rangle=0.
\end{equation}
Equivalently,
$\mathrm dS_t^n/S_t^n
=\mathrm dJ_t^n+\mathrm dN_t^n$.
Set
\[
 C^n=\langle J^n\rangle,
 \qquad D^n=\langle N^n\rangle.
\]
There is a $\sigma$-field $\mathscr H^n$ such that the paths of
$J^n,C^n,D^n$ up to $T_n$ are measurable with respect to $\mathscr H^n$ and
conditionally on $\mathscr H^n$,
\begin{equation}\label{eq:conditional-time-change}
 \{N_t^n\}_{0\le t\le T_n}
 \ \stackrel{d}{=}\ \{B_{D_t^n}^n\}_{0\le t\le T_n},
\end{equation}
where $B^n$ is a standard Brownian motion independent of
$\mathscr H^n$.  Finally, for a constant
$\underline c\in(0,1]$ independent of $n$,
\begin{equation}\label{eq:clock-nondegeneracy}
 D_{T_n}^n\ge
 \underline c A^n\qquad\Q\text{-a.s.}
\end{equation}
\end{assumption}

Here $A_t^n=C_t^n+D_t^n$.  Equivalently,
conditionally on $\mathscr H^n$, $N^n$ is a centered
Gaussian process with independent increments and variance clock
$D^n$.  The Brownian motion in
\eqref{eq:conditional-time-change} must be independent of the whole
environment $\mathscr H^n$, not only of $D^n$.  The
clock itself may depend on $J^n$.  This is the usual situation
in stochastic-volatility models.

\begin{example}[Brownian-factor stochastic volatility]
Let $(B,B^\perp)$ be a two-dimensional standard Brownian motion and
let $\{\mathscr G_t\}$ be the filtration generated by the first component $B$.
A stochastic volatility model has the form 
\begin{equation*}
 M^n_t = \int_0^t \sqrt{V^n_s}\,\mathrm{d}W_s, \quad  W = \rho B + \sqrt{1-\rho^2} B^\perp
\end{equation*}
 for a positive continuous $\{\mathscr G_t\}$-adapted process $V^n$ and $\rho \in (-1,1)$.
Then, Assumption~\ref{ass:smoothing} is met with
\begin{equation*}
 \mathscr H^n = \mathscr G_{T_n}, \quad
J^n_t = \rho \int_0^t \sqrt{V^n_s}\,\mathrm{d}B_s, \quad
N^n_t = \sqrt{1-\rho^2} \int_0^t \sqrt{V^n_s}\,\mathrm{d}B^\perp_s, \quad
\underline c = 1 - \rho^2
\end{equation*}
because, conditionally on $\mathscr G_{T_n}$, the integrand in
$N^n$ is fixed and $B^\perp$ remains an independent Brownian motion.
Thus $N^n$ is conditionally Gaussian with independent increments
and variance clock $(1-\rho^2)\int_0^t V_s^n\,\mathrm ds$.
\end{example}

\subsection{Main results}
\label{subsec:main-results}

Let $k_n\in\mathbb R$, set
$K_n=S_0e^{k_n}$, and denote by
$\widehat w^n(k_n,T_n)$ the total implied
variance defined by \eqref{eq:implied-total-variance} for the model
$S^n$.  Define
\begin{equation}\label{eq:standardized-strike}
 z_n=\frac{k_n+v_n/2}{\sqrt{v_n}}.
\end{equation}

\begin{theorem}[Conditional quadratic-variation formula]
\label{thm:main}
Suppose Assumptions~\ref{ass:qv} and \ref{ass:smoothing} hold.  If
$\{z_n\}$ is bounded, then the total implied variance
of the put with maturity $T_n$ and strike $K_n$ satisfies
\begin{equation}\label{eq:main-result}
 \widehat w^n(k_n,T_n)
 =\E\!\left[A^n\mid
       S_{T_n}^n=K_n\right]
  +o(v_n r_n).
\end{equation}
The conditional expectation is the continuous density version defined
in \eqref{eq:conditional-A} below.
\end{theorem}

The statement in annualized units is
\begin{equation}\label{eq:annualized-result}
 \widehat\sigma^n(k_n,T_n)^2
 =\frac1{T_n}
   \E\!\left[A^n\mid
       S_{T_n}^n=K_n\right]
  +o\!\left(\frac{v_n r_n}{T_n}\right).
\end{equation}

\paragraph{Comparison with the martingale expansion of Fukasawa (2026).}
The fixed-maturity predecessor of Theorem~\ref{thm:main} expands
implied variance through a limiting regression.  We state its
implied-variance consequence using a sequence $a_n\to0$ as the
absolute perturbation scale, to distinguish it from the relative
scale $r_n$ of Assumption~\ref{ass:qv}.

\begin{theorem}[Fukasawa~{\cite[Theorem~2.1, Corollary~2.4 and Remark~2.5]{Fukasawa2026}}]
\label{thm:fukasawa2026}
Fix $T>0$.  Let
\[
 \frac{\mathrm dS_t^n}{S_t^n}
 =\sqrt{V_t^n}\bigl(\rho\,\mathrm dW_t
          +\sqrt{1-\rho^2}\,\mathrm dW_t^\perp\bigr),
 \qquad |\rho|<1,
\]
where $V^n$ is nonnegative and c\`adl\`ag, adapted to a filtration
to which $W$ is adapted and from which $W^\perp$ is independent.
Here $(W,W^\perp)$ is a two-dimensional Brownian motion.
Suppose $v_n=\E[A^n]\to v_*>0$ and $a_n>0$ tends to zero.
Set
\[
 X^n=\frac1{\sqrt{v_n}}\int_0^T\frac{\mathrm dS_t^n}{S_t^n},
 \qquad R^n=\frac{A^n-v_n}{a_n},
 \qquad A^n=\int_0^T V_t^n\,\mathrm dt.
\]
Assume $\{R^n\}$ is uniformly integrable and
$(X^n,R^n)\Longrightarrow(X,R)$.
Let $b(x):=\E[R\mid X=x]$ denote the smooth version described below.
Then, for each fixed $K=S_0e^k>0$,
\begin{equation}\label{eq:fukasawa2026-expansion}
 \widehat w^n(k,T)=v_n+a_n b(z_n)+o(a_n),
 \qquad z_n=\frac{k+v_n/2}{\sqrt{v_n}}.
\end{equation}
\end{theorem}

Here $X$ is standard normal.  If $(X,R)$ is centered Gaussian,
then $b(z)=\E[XR]z$, giving an explicit linear correction.
The centering $v_n=\E[A^n]$ is the specialization relevant here;
the cited result also permits deterministic centerings with a
strictly positive limit.

Gaussian smoothing also explains the regularity of the limiting
regression.  Along a subsequence the normalized correlated
martingale and $R^n$ have a joint limit $(U,R)$, and
$X=U+\sqrt d\,G$ in law jointly with $R$, where
$d=1-\rho^2>0$ and $G$ is standard normal independent of $(U,R)$.
Consequently,
\[
 q(x):=b(x)\phi(x)
 =\E\!\left[\frac{R}{\sqrt d}
       \phi\!\left(\frac{x-U}{\sqrt d}\right)\right].
\]
Gaussian differentiation and $\E|R|<\infty$ imply that $q,q',q''$
are integrable and that $q,q'\to0$ at both infinities.
In fact, every derivative of the Gaussian kernel is bounded, so
differentiation under the expectation gives $q\in C^\infty$.
Since $\phi>0$, the ratio $b=q/\phi$ defines a $C^\infty$ version
of the regression.  No separate differentiability assumption on
$b$ is therefore needed in Theorem~\ref{thm:fukasawa2026}.

To see the formal connection with our formula, define
\[
 \widetilde X^n
 =\frac{\log(S_T^n/S_0)+v_n/2}{\sqrt{v_n}}
 =X^n-\frac{a_nR^n}{2\sqrt{v_n}}.
\]
The pairs $(\widetilde X^n,R^n)$ and $(X^n,R^n)$ have the same
weak limit, and conditioning on $\widetilde X^n=z_n$ is exactly
conditioning on $S_T^n=K$.  Formally replacing the limiting
regression $b(z_n)$ in \eqref{eq:fukasawa2026-expansion} by
$\E[R^n\mid\widetilde X^n=z_n]$ therefore gives
\[
 v_n+a_n\E[R^n\mid S_T^n=K]
 =\E[A^n\mid S_T^n=K].
\]
This is the substitution proposed in
\cite[Remark~2.5]{Fukasawa2026}.  Weak convergence alone does not
justify convergence of these regressions at a specified point.

Our theorem validates the conditional-variance representation
under a different set of assumptions.  With $r_n=a_n/v_n$, its
weighted perturbation is
\[
 Y^n=\frac{R^n}{\sqrt{Z^n}},\qquad Z^n=A^n/v_n.
\]
Uniform integrability of this weighted family is an additional
small-variance requirement relative to uniform integrability of
$R^n$ alone.  Since $v_n\to v_*>0$, it also implies our tail
condition, and the error $o(v_nr_n)$ becomes $o(a_n)$.
No joint limit or differentiability of a limiting regression is
needed for Theorem~\ref{thm:main}.  It also allows varying
maturities and continuous clocks without spot variance.
Whenever both theorems apply, their conclusions imply
\[
 \E[R^n\mid S_T^n=K]-b(z_n)\longrightarrow0.
\]
Thus the new formula retains the prelimit conditional expectation
whose limiting coefficient appears in Fukasawa~\cite{Fukasawa2026}.

\subsection{Short maturity asymptotics}
\label{subsec:short-maturity}

We first specialize Theorem~\ref{thm:main} to vanishing maturities,
then relate the result to the short-time expansion in
Fukasawa~\cite{Fukasawa2021Rough}.  A spot-variance moment criterion and a
Gaussian multifactor model give concrete sufficient conditions.

\begin{corollary}[Short maturity at the fluctuation scale]
\label{cor:short-time}
Let $T_n\to0$ and let $a(T)>0$ satisfy $a(T)\to0$ as
$T\downarrow0$.  Write $a_n=a(T_n)$, set
\begin{equation}\label{eq:average-qv-perturbation}
 Y^n
 :=\frac{Z^n-1}{a_n\sqrt{Z^n}},
\end{equation}
and suppose that
\[
 0<\underline v
 \le\frac{v_n}{T_n}
 \le\overline v<\infty.
\]
Suppose that $\{Y^n\}$ is uniformly integrable, that
Assumption~\ref{ass:smoothing} holds, and that, for every $\eta>0$,
\begin{equation}\label{eq:short-time-corollary-tail}
 \Q(|Z^n-1|>\eta)
 =o(a_n\sqrt{T_n}).
\end{equation}
If $k_n=O(\sqrt{T_n})$, then
\begin{align}
 \widehat\sigma^n(k_n,T_n)^2
 &=\E\!\left[
       \frac{A^n}{T_n}
       \mathrel{\Big|}S_{T_n}^n=K_n
      \right]+o(a_n).
 \label{eq:short-time-result}\\
 \widehat w^n(k_n,T_n)
 &=\E\!\left[
       A^n
       \mathrel{\Big|}S_{T_n}^n=K_n
      \right]+o(a_n T_n).
 \label{eq:short-time-total-result}
\end{align}
In particular, for $H\in(0,1/2)$ and $a(T)=T^H$,
\begin{equation}\label{eq:rough-short-time-result}
 \widehat\sigma^n(k_n,T_n)^2
 =\E\!\left[
       \frac{A^n}{T_n}
       \mathrel{\Big|}S_{T_n}^n=K_n
      \right]+o(T_n^H).
\end{equation}
In this specialization, \eqref{eq:short-time-corollary-tail} reads
\[
 \Q(|Z^n-1|>\eta)=o(T_n^{H+1/2}).
\]
\end{corollary}

\begin{proof}[Proof of Corollary~\ref{cor:short-time}]
Apply Theorem~\ref{thm:main} with
$r_n=a_n$.  The bounds on
$v_n/T_n$ give
\[
 r_n(1+\sqrt{v_n})=O(a_n)\longrightarrow0,
 \qquad
 r_n\sqrt{v_n}\asymp
 a_n\sqrt{T_n}.
\]
They also show that $k_n=O(\sqrt{T_n})$ makes
$z_n$ bounded.  Thus
\eqref{eq:short-time-corollary-tail} is precisely the required
localization condition, up to fixed multiplicative constants.
\end{proof}

\paragraph{Connection with the short-time expansion.}
We state the expansion from \cite{Fukasawa2021Rough}
in our notation before comparing it with
Corollary~\ref{cor:short-time}.  For a fixed model with spot
variance $V$, write
\[
 A_T=\int_0^T V_t\,\mathrm dt,\qquad
 m(t)=\E[V_t],\qquad v_T=\int_0^T m(t)\,\mathrm dt,
 \qquad Z_T=A_T/v_T.
\]
\begin{theorem}[Fukasawa~{\cite[Theorem~2.1]{Fukasawa2021Rough}}]
\label{thm:fukasawa2021}
Suppose that $S$ is a strictly positive continuous martingale,
$\mathrm d\langle\log S\rangle_t=V_t\,\mathrm dt$, and $m$ is
positive and continuous at zero.  For some $H\in(0,1/2]$, suppose
that
\[
 U_t:=t^{-H}\left(\frac{V_t}{m(t)}-1\right)
\]
is uniformly integrable as $t\downarrow0$ and
\[
 \left(\frac{S_t/S_0-1}{\sqrt t},U_t\right)
 \Longrightarrow(\xi,\eta).
\]
Put $v_0=m(0)$, $\bar v_T=v_T/T$, and
$g(a)=\E[\eta\mid\xi=a]$.  For each fixed $x\in\mathbb R$ and
every sequence $T_n\downarrow0$,
\begin{equation}\label{eq:fukasawa2021-expansion}
 \widehat\sigma(x\sqrt{T_n},T_n)
 =\sqrt{\bar v_{T_n}}\{1+\alpha(x)T_n^H\}+o(T_n^H),
\end{equation}
with
\begin{equation}\label{eq:rough-alpha-bridge}
 2\alpha(x)=\int_0^1 u^H\int_{\mathbb R}
 g\!\left(x\sqrt u+\sqrt{v_0(1-u)}\,y\right)
\phi(y)\,\mathrm dy\,\mathrm du.
\end{equation}
\end{theorem}

The coefficient uses the limiting regression of spot variance on
the return.  Corollary~\ref{cor:short-time} instead retains the
conditional expectation of integrated variance at the actual
maturity.  It requires smoothing and quantitative tail control,
but neither a joint limit nor the true-martingale assumption in
Theorem~\ref{thm:fukasawa2021}.

To connect this formally with our formula, freeze the return process
at $\sqrt{v_0}B$ on the rescaled interval $[0,1]$.
For $0<u<1$, the Brownian bridge identity is
\[
 \mathcal L\!\left(\frac{\sqrt{v_0}B_u}{\sqrt u}
           \mathrel{\Big|}\sqrt{v_0}B_1=x\right)
 =\mathcal N(x\sqrt u,v_0(1-u)).
\]
Replacing the conditional regression of $U_{Tu}$ by $g$, and using
independence of the future Brownian increment from the information
at time $u$, therefore suggests
\[
 \E[U_{Tu}\mid S_T=S_0e^{x\sqrt T}]
 \approx\int_{\mathbb R}
 g\!\left(x\sqrt u+\sqrt{v_0(1-u)}\,y\right)\phi(y)\,\mathrm dy.
\]
Since $V_{Tu}=m(Tu)\{1+T^H u^H U_{Tu}\}$, integration gives
\[
 \E\!\left[\frac{A_T}{T}\mathrel{\Big|}S_T=S_0e^{x\sqrt T}\right]
 =\bar v_T+2v_0\alpha(x)T^H+o(T^H)
 =\widehat\sigma(x\sqrt T,T)^2+o(T^H),
\]
formally recovering the new formula.  Passing to conditional
expectations at a fixed terminal value and integrating the limit
are the formal steps here; joint weak convergence alone does not
justify them.  The smoothing and integrability conditions of the
present theorem provide a rigorous route to the representation.

\paragraph{Sufficient spot-variance conditions.}
The averaging identity
\begin{equation}\label{eq:rough-averaging}
 \frac{Z_T-1}{T^H}
 =\int_0^T\frac{m(t)}{v_T}
       \left(\frac tT\right)^H U_t\,\mathrm dt
\end{equation}
transfers uniform integrability to the averaged perturbation
$(Z_T-1)/T^H$: apply Jensen's inequality to a convex
superlinear function witnessing uniform integrability, with value
zero at the origin.  This alone does not control the additional
factor $Z_T^{-1/2}$ in $Y_T$, or the quantitative tail condition
\eqref{eq:short-time-localization}.
The following moment strengthening gives both and verifies the
hypotheses of Corollary~\ref{cor:short-time} when
Assumption~\ref{ass:smoothing} also holds.

\begin{proposition}\label{prop:rough-spot-moments}
Suppose that $V_t>0$, that $m$ is positive and continuous at zero,
and that, for some $p>1+1/(2H)$ and $T_0>0$,
\begin{equation}\label{eq:rough-spot-moments}
 \sup_{0<t\le T_0}\E[|U_t|^{2p}]<\infty,
 \qquad
 \sup_{0<t\le T_0}
 \E\!\left[\left(\frac{V_t}{m(t)}\right)^{-p}\right]<\infty.
\end{equation}
Then Assumption~\ref{ass:qv} holds with $r_n=T_n^H$ for
every sequence $T_n\downarrow0$ with $T_n\le T_0$.
\end{proposition}

\begin{proof}
By \eqref{eq:rough-averaging} and Minkowski's inequality,
$\|Z_T-1\|_{L^{2p}}=O(T^H)$.  Jensen's inequality with
probability weights $m(t)\,\mathrm dt/v_T$ gives
\[
 \E[Z_T^{-p}]
 \le\int_0^T\frac{m(t)}{v_T}
       \E\!\left[\left(\frac{V_t}{m(t)}\right)^{-p}\right]
       \,\mathrm dt=O(1).
\]
Consequently, H\"older's inequality yields
\[
 \E\left[\left|\frac{Z_T-1}{T^H\sqrt{Z_T}}\right|^p\right]
 \le T^{-Hp}\|Z_T-1\|_{L^{2p}}^p
          \E[Z_T^{-p}]^{1/2}=O(1).
\]
Since $v_T\sim m(0)T$ and
$T^{H(p-1)-1/2}\to0$, Proposition~\ref{prop:Lp} applies.
\end{proof}

\begin{example}[Gaussian multifactor volatility]
\label{ex:gaussian-multifactor}
Let $B^1,\ldots,B^d,B^\perp$ be independent standard Brownian
motions and let
\[
 X_t^i=\int_0^t K_i(t,s)\,\mathrm dB_s^i,\qquad i=1,\ldots,d,
\]
where the deterministic kernels give continuous adapted Gaussian
processes and, for some $H\in(0,1/2]$,
\begin{equation}\label{eq:gaussian-kernel-bound}
 \int_0^t K_i(t,s)^2\,\mathrm ds\le C t^{2H}.
\end{equation}
Take $V_t=f(t,X_t^1,\ldots,X_t^d)$, where $f>0$ is continuous in
$(t,x)$, twice continuously differentiable in $x$, and its spatial
derivatives are jointly continuous.  Suppose that, for some $C,a>0$,
\begin{equation}\label{eq:gaussian-f-growth}
 f(t,x)+f(t,x)^{-1}+|\nabla_x f(t,x)|
       +\|\nabla_x^2 f(t,x)\|\le Ce^{a|x|},
 \qquad 0\le t\le T_0.
\end{equation}
No differentiability in time is required.  The reciprocal bound
controls small variances while allowing $f$ to approach zero.
Set
\begin{equation}\label{eq:gaussian-multifactor-price}
 \frac{\mathrm dS_t}{S_t}
 =\sqrt{V_t}\left(\sum_{i=1}^d\rho_i\,\mathrm dB_t^i
       +\sqrt{1-|\rho|^2}\,\mathrm dB_t^\perp\right),
 \qquad |\rho|^2:=\sum_{i=1}^d\rho_i^2<1.
\end{equation}
Gaussian exponential moments and \eqref{eq:gaussian-f-growth} imply, for every
finite $q\ge1$,
\[
 \|V_t-f(t,0)\|_{L^q}=O(t^H),\qquad
 \|V_t-m(t)\|_{L^q}=O(t^H),\qquad m(t)\longrightarrow f(0,0).
\]
Indeed, the spatial mean-value estimate is bounded by
$C|X_t|e^{a|X_t|}$, whose $L^q$ norm is $O(t^H)$.
The reciprocal bound likewise gives
\[
 \sup_{0<t\le T_0}\E[V_t^{-q}]<\infty,
 \qquad
 \sup_{0<t\le T_0}\E[V_t^q]<\infty.
\]
After decreasing $T_0$ if necessary, $m(t)$ is bounded above and
away from zero.
Thus \eqref{eq:rough-spot-moments} holds for every finite $p$, and
Proposition~\ref{prop:rough-spot-moments} verifies
Assumption~\ref{ass:qv} with $r_n=T_n^H$.
For smoothing, take
\[
 J_t=\sum_{i=1}^d\rho_i\int_0^t\sqrt{V_s}\,\mathrm dB_s^i,
 \qquad
 N_t=\sqrt{1-|\rho|^2}\int_0^t\sqrt{V_s}\,\mathrm dB_s^\perp.
\]
Conditionally on the factor Brownian paths, $N$ is Gaussian with
independent increments and $D_T=(1-|\rho|^2)A_T$.
Consequently, for $k_n=O(\sqrt{T_n})$,
\begin{equation}\label{eq:gaussian-multifactor-conditional}
 \widehat\sigma(k_n,T_n)^2
 =\E\!\left[\frac{A_{T_n}}{T_n}
       \mathrel{\Big|}S_{T_n}=S_0e^{k_n}\right]+o(T_n^H).
\end{equation}

\end{example}

\paragraph{Gaussian coefficient and volatility skew.}
In Example~\ref{ex:gaussian-multifactor}, an additional kernel limit
identifies the coefficient explicitly.  Suppose that
\begin{equation}\label{eq:gaussian-kernel-limit}
 t^{1/2-H}K_i(t,tu)\longrightarrow\kappa_i(u)
 \quad\hbox{in }L^2(0,1).
\end{equation}
Write $v_0=f(0,0)$ and $f_i=\partial_i f(0,0)$.
Define
\begin{equation}\label{eq:gaussian-cross-covariance}
 \Sigma_{12}:=\frac1{\sqrt{v_0}}
   \sum_{i=1}^d\rho_i f_i\int_0^1\kappa_i(u)\,\mathrm du.
\end{equation}
Then, for fixed $x\in\mathbb R$ and $\bar v_T=v_T/T$,
\begin{equation}\label{eq:rough-gaussian-skew}
 \widehat\sigma(x\sqrt{T_n},T_n)
 =\sqrt{\bar v_{T_n}}
  +\frac{\Sigma_{12}x}{2\sqrt{v_0}(H+3/2)}T_n^H
  +o(T_n^H).
\end{equation}
This is the Gaussian coefficient of
\cite[Corollary~2.1]{Fukasawa2021Rough}.  The deterministic time
dependence is retained in $\bar v_T$ and need not have a power-law
expansion.

\begin{proof}[Derivation of \eqref{eq:rough-gaussian-skew}]
Taylor's formula and Gaussian moments give
\[
 V_t=f(t,0)+\sum_i\partial_i f(t,0)X_t^i+O_{L^2}(t^{2H}),
 \qquad m(t)=f(t,0)+O(t^{2H}).
\]
Thus $V_t-m(t)=\sum_i f_iX_t^i+o_{L^2}(t^H)$.
The log return divided by $\sqrt t$ differs in $L^2$ by $o(1)$ from
$\sqrt{v_0}(\sum_i\rho_iB_t^i+
\sqrt{1-|\rho|^2}B_t^\perp)/\sqrt t$.
Replacing $\log(S_t/S_0)$ by $S_t/S_0-1$ changes the normalized
quantity by $o_{\Q}(1)$.
It follows from \eqref{eq:gaussian-kernel-limit} that
$((S_t/S_0-1)/\sqrt t,U_t)$ converges to a centered Gaussian
pair $(\xi,\eta)$ with $\E[\xi^2]=v_0$ and
$\E[\xi\eta]=\Sigma_{12}$.
To evaluate the terminal conditional expectation directly, set
\[
 L_T=\frac1{T^{H+1}}\sum_i f_i\int_0^T X_t^i\,\mathrm dt,
 \qquad R_T^0=\frac{\sqrt{v_0}}{\sqrt T}
    \left(\sum_i\rho_iB_T^i+\sqrt{1-|\rho|^2}B_T^\perp\right).
\]
Then $T^{-H}(A_T/T-\bar v_T)=L_T+o_{L^2}(1)$,
$\sup_T\E[L_T^2]<\infty$, and kernel scaling gives
\[
 \operatorname{Cov}(L_T,R_T^0)
 \longrightarrow\frac{v_0\Sigma_{12}}{H+3/2},
 \qquad \operatorname{Var}(R_T^0)=v_0.
\]
Put $d_0=(1-|\rho|^2)v_0$ and define the density of $R_T^0$
conditional on the factor Brownian paths by
\[
 H_T^0(x)=\frac1{\sqrt{d_0}}
 \phi\!\left(\frac{x-\sqrt{v_0}\sum_i\rho_iB_T^i/\sqrt T}
                  {\sqrt{d_0}}\right).
\]
Writing $h_T$ for the kernel in \eqref{eq:conditional-kernel},
we have $\sqrt T h_T(x\sqrt T)-H_T^0(x)\to0$ in $L^2$.
Indeed, $J_T/\sqrt T$ approaches its frozen Gaussian counterpart,
$D_T/T\to(1-|\rho|^2)v_0$, and $A_T/\sqrt T\to0$ in probability;
the inverse moments above and the kernel bound
\eqref{eq:kernel-bound} give uniform integrability of every fixed
power of the scaled density.  The denominator tends to the
positive $\mathcal N(0,v_0)$ density at $x$.
Moreover, Gaussian regression gives the exact identity
\[
 \frac{\E[L_TH_T^0(x)]}{\E[H_T^0(x)]}
 =\frac{\operatorname{Cov}(L_T,R_T^0)}{v_0}x.
\]
The $L^2$ bounds justify the numerator replacements in
\eqref{eq:conditional-A}, yielding
\begin{equation}\label{eq:rough-gaussian-conditional-coefficient}
 \E\!\left[\frac{A_{T_n}}{T_n}
       \mathrel{\Big|}S_{T_n}=S_0e^{x\sqrt{T_n}}\right]
 =\bar v_{T_n}
   +\frac{\Sigma_{12}x}{H+3/2}T_n^H+o(T_n^H).
\end{equation}
Taking the square root in
\eqref{eq:gaussian-multifactor-conditional} proves
\eqref{eq:rough-gaussian-skew}.  This argument uses only the
local martingale property of $S$.
\end{proof}

For distinct fixed $x,y$, \eqref{eq:rough-gaussian-skew} gives
\[
 \frac{\widehat\sigma(x\sqrt{T_n},T_n)
       -\widehat\sigma(y\sqrt{T_n},T_n)}{(x-y)\sqrt{T_n}}
 =\frac{\Sigma_{12}}{\sqrt{v_0}(2H+3)}T_n^{H-1/2}
       +o(T_n^{H-1/2}).
\]
This is a finite-difference skew statement; differentiation of the
remainder is not needed.
For Ornstein--Uhlenbeck factors,
$K_i(t,s)=\nu_i e^{-\lambda_i(t-s)}$ with fixed
$\lambda_i>0$, one has $H=1/2$ and $\kappa_i=\nu_i$.
A fixed finite collection of these factors therefore has a finite
limiting skew.  Kernels
$K_i(t,s)=\nu_i(t-s)^{H-1/2}$ with $H<1/2$ instead yield
$\kappa_i(u)=\nu_i(1-u)^{H-1/2}$ and a skew of order
$T_n^{H-1/2}$ when $\Sigma_{12}\ne0$.
The distinction agrees with the roughness mechanism discussed in
\cite{Fukasawa2021Rough}; a finite set of fixed mean-reversion rates
cannot produce that blow-up in the limit $T_n\to0$.

\subsection{Small volatility-of-volatility}
\label{subsec:small-vov}

Small volatility-of-volatility expansions are regular perturbations
around a deterministic variance path.  Lewis~\cite[Chapter~3]{Lewis2000}
develops a power-series approach to option valuation in this regime under the Heston model.
Benhamou, Gobet and Miri~\cite{BGM2010} derive expansions with error
estimates for the time-dependent Heston model using Malliavin calculus.
Fukasawa~\cite{Fukasawa2011Martingale,Fukasawa2026} gives martingale
expansion frameworks covering this regime.
Here we verify the conditional quadratic-variation formula directly
from SDE moment estimates.

Fix $T_n=T>0$ and let $a_n>0$ tend to zero.  Consider the general
stochastic-volatility SDE
\begin{equation}\label{eq:small-vov-SDE}
 \begin{aligned}
  \frac{\mathrm dV_t^n}{V^n_t}
 &=b(t,V_t^n)\,\mathrm dt+a_n c(t,V_t^n)\,\mathrm dW_t,
 & V_0^n&=v_*>0,\\
 \frac{\mathrm dS_t^n}{S_t^n}
 &=\sqrt{V_t^n}\left(\rho\,\mathrm dW_t
       +\sqrt{1-\rho^2}\,\mathrm dW_t^\perp\right),
 & S_0^n&=S_0,
 \end{aligned}
\end{equation}
where $W,W^\perp$ are independent standard Brownian motions and
$|\rho|<1$.  The coefficient $a_nc$ controls the
volatility-of-volatility.  Its deterministic zero-noise limit solves
\[
 \mathrm d\bar V_t=\bar V_t b(t,\bar V_t)\,\mathrm dt,
 \qquad \bar V_0=v_*,\qquad
 \bar A=\int_0^T\bar V_t\,\mathrm dt.
\]

Suppose that $b,c$ are bounded on $[0,T]\times(0,\infty)$,
continuous in $t$, and uniformly Lipschitz in their second argument.
The coefficients $vb(t,v)$ and $vc(t,v)$ are locally Lipschitz
with linear growth.  The variance SDE therefore has a unique
nonexplosive strong solution, and its exponential representation
\[
 V_t^n=v_*\exp\!\left\{
   \int_0^t\left(b(s,V_s^n)-\frac{a_n^2}{2}c(s,V_s^n)^2\right)
       \,\mathrm ds
   +a_n\int_0^t c(s,V_s^n)\,\mathrm dW_s\right\}
\]
shows that it stays strictly positive.  More precisely, apply this
identity up to exits from compact subsets of $(0,\infty)$;
bounded $b,c$ prevent the logarithm from diverging in finite time.
Since $\{a_n\}$ is bounded, exponential martingale estimates give,
for every finite $q\ge1$,
\begin{equation}\label{eq:small-vov-positive-negative-moments}
 \sup_n\E\!\left[\sup_{0\le t\le T}
       \bigl((V_t^n)^q+(V_t^n)^{-q}\bigr)\right]<\infty.
\end{equation}
Indeed, the martingale in the exponent has quadratic variation
bounded by $a_n^2T\|c\|_\infty^2$, and its drift is uniformly
bounded.  The deterministic solution satisfies
$v_*e^{-T\|b\|_\infty}\le\bar V_t
\le v_*e^{T\|b\|_\infty}$.

The Burkholder--Davis--Gundy and Gronwall inequalities give, for every
finite $p\ge2$,
\begin{equation}\label{eq:small-vov-moment}
 \left\|\sup_{0\le t\le T}|V_t^n-\bar V_t|\right\|_{L^p}
 \le C_{p,T}a_n.
\end{equation}
To justify the drift estimate despite the lack of global
Lipschitz continuity of $v\mapsto vb(t,v)$, use
\[
 vb(t,v)-\bar V_t b(t,\bar V_t)
 =(v-\bar V_t)b(t,v)
   +\bar V_t\{b(t,v)-b(t,\bar V_t)\}.
\]
Its absolute value is at most $C_T|v-\bar V_t|$.
The diffusion term is controlled by
\[
 \left\|\sup_{t\le T}
    \left|a_n\int_0^t V_s^n c(s,V_s^n)\,\mathrm dW_s\right|
      \right\|_{L^p}
 \le C_p a_n\|c\|_\infty
    \left\|\left(\int_0^T(V_s^n)^2\,\mathrm ds\right)^{1/2}
      \right\|_{L^p}
 =O(a_n)
\]
by \eqref{eq:small-vov-positive-negative-moments}.
Gronwall's inequality gives \eqref{eq:small-vov-moment}.
Consequently,
\[
 \|A^n-\bar A\|_{L^p}=O(a_n),\qquad
 |v_n-\bar A|=O(a_n),\qquad
 \|A^n-v_n\|_{L^p}=O(a_n).
\]
To control the inverse quadratic variation, Jensen's inequality
and \eqref{eq:small-vov-positive-negative-moments} give
\[
 \E[(A^n)^{-2}]
 \le T^{-3}\int_0^T\E[(V_t^n)^{-2}]\,\mathrm dt\le C_T.
\]
Since $v_n\to\bar A>0$, taking $r_n=a_n$ and applying
H\"older's inequality yields
\[
 \E[|Y^n|^2]
 =\frac1{a_n^2v_n}\E\!\left[\frac{(A^n-v_n)^2}{A^n}\right]
 \le\frac{\|A^n-v_n\|_{L^4}^2}{a_n^2v_n}
       \E[(A^n)^{-2}]^{1/2}=O(1).
\]
Thus Proposition~\ref{prop:Lp}, with $p=2$, verifies
Assumption~\ref{ass:qv}; its rate condition is
$a_n/\sqrt{v_n}\to0$.
Moreover, setting
\[
 J_t^n=\rho\int_0^t\sqrt{V_s^n}\,\mathrm dW_s,
 \qquad
 N_t^n=\sqrt{1-\rho^2}\int_0^t\sqrt{V_s^n}\,\mathrm dW_s^\perp,
\]
conditioning on $\sigma(W_s:0\le s\le T)$ verifies
Assumption~\ref{ass:smoothing} with $\underline c=1-\rho^2$.
For any bounded sequence $\{k_n\}$ and $K_n=S_0e^{k_n}$,
Theorem~\ref{thm:main} therefore yields
\begin{align}
 \widehat w^n(k_n,T)
 &=\E[A^n\mid S_T^n=K_n]+o(a_n),
 \label{eq:small-vov-total}\\
 \widehat\sigma^n(k_n,T)^2
 &=\frac1T\E[A^n\mid S_T^n=K_n]+o(a_n).
 \label{eq:small-vov-annualized}
\end{align}

The coefficient conditions above can be replaced by direct moment
estimates.  The same conclusion holds for positive strong solutions of
\eqref{eq:small-vov-SDE}, with $V^n$ adapted to the filtration generated
by $W$, whenever $v_n\to\bar A>0$,
$\|A^n-v_n\|_{L^4}=O(a_n)$, and
$\sup_n\E[(A^n)^{-2}]<\infty$, by the same argument.

\subsection{Fast mean reverting asymptotics}
\label{subsec:fast-mr}

Fast mean-reverting stochastic volatility leads to singular
perturbation expansions.  Fouque, Papanicolaou and
Sircar~\cite{FPS2000} develop this approach to option pricing;
Fouque, Papanicolaou, Sircar and S\o lna~\cite{FPSS2003} establish
error estimates for call options with their nonsmooth payoff.
Their monograph~\cite{FPSS2011} treats multiscale stochastic volatility
and its applications.  Khasminskii and Yin~\cite{KY2005} obtain uniform
asymptotic expansions and error bounds for fast diffusion models.
Fukasawa~\cite{Fukasawa2011Edgeworth} validates the expansion through
Edgeworth theory for ergodic diffusions.

Following the time-scale interpretation in
Fukasawa~\cite[Section~2.1]{Fukasawa2011Edgeworth}, start with a general
one-dimensional ergodic diffusion
\[
 \mathrm dX_u=b(X_u)\,\mathrm du+c(X_u)\,\mathrm dW_u
\]
on an interval $I$, with invariant probability measure $\pi$.
Assume that this SDE has a nonexplosive strong solution.  Let
$W^\perp$ be a Brownian motion independent of $(X_0,W)$, and fix
$T_n=T>0$.  Define the accelerated factor and Brownian motions by
\[
 X_t^n=X_{nt},\qquad
 W_t^n=n^{-1/2}W_{nt},\qquad
 W_t^{\perp,n}=n^{-1/2}W_{nt}^\perp.
\]
Then $W^n,W^{\perp,n}$ are independent standard Brownian motions in
the accelerated filtration, and
\begin{equation}\label{eq:fast-mr-SDE}
 \mathrm dX_t^n=nb(X_t^n)\,\mathrm dt
                 +\sqrt n\,c(X_t^n)\,\mathrm dW_t^n.
\end{equation}
For a positive variance function $f$ with
$\bar v:=\int_I f\,\mathrm d\pi\in(0,\infty)$, set
\begin{equation}\label{eq:fast-mr-price}
 \frac{\mathrm dS_t^n}{S_t^n}
 =\sqrt{f(X_t^n)}\left(
     \rho(X_t^n)\,\mathrm dW_t^n
     +\sqrt{1-\rho(X_t^n)^2}\,\mathrm dW_t^{\perp,n}\right),
 \qquad S_0^n=S_0.
\end{equation}
Here $1-\rho(x)^2\ge\underline c>0$ on $I$.
The stochastic exponential defining $S^n$ is a strictly positive
continuous local martingale.  The quadratic variation of its
stochastic logarithm is integrable under the moment conditions below.

The change of time variable gives
\begin{equation}\label{eq:fast-mr-time-change}
 A^n=\int_0^Tf(X_{nt})\,\mathrm dt
     =\frac1n\int_0^{nT}f(X_u)\,\mathrm du
     \longrightarrow T\bar v
\end{equation}
by the ergodic theorem.  Thus fast mean reversion is a long-time
averaging problem for $X$.  To obtain a first-order remainder, we
also need moment control of the centered additive functional.
The following sufficient conditions implement the martingale
representation used in \cite[Section~2.2]{Fukasawa2011Edgeworth}.

Let $g\in C^2(I)$ solve the Poisson equation
\begin{equation}\label{eq:fast-mr-Poisson}
 \mathcal Lg=f-\bar v,\qquad
 \mathcal Lg=bg'+\frac12c^2g'',\qquad h=cg',
\end{equation}
and suppose that
\begin{equation}\label{eq:fast-mr-moments}
 \sup_{u\ge0}\E\!\left[
      |g(X_u)|^4+|h(X_u)|^4+f(X_u)^{-2}\right]<\infty.
\end{equation}
These are conditions on the unscaled diffusion.  They do not require
stationary initialization or geometric mixing.  If $X_0\sim\pi$,
it suffices that $g,h\in L^4(\pi)$ and $f^{-1}\in L^2(\pi)$.
For a regular diffusion with invariant density $p$ and zero stationary
probability current, a solution can be constructed, whenever the
integrals and regularity permit, from
\[
 g'(x)=\frac{2}{c(x)^2p(x)}
          \int_{\ell}^{x}(f(y)-\bar v)p(y)\,\mathrm dy,
 \qquad \ell=\inf I.
\]
The moment condition \eqref{eq:fast-mr-moments} must still be checked
for this solution.

It\^o's formula yields the exact decomposition
\begin{equation}\label{eq:fast-mr-additive}
 A^n-T\bar v
 =\frac{g(X_{nT})-g(X_0)}{n}
   -\frac1n\int_0^{nT}h(X_u)\,\mathrm dW_u.
\end{equation}
The first term is $O(n^{-1})$ in $L^4$.
By the Burkholder--Davis--Gundy inequality and
\eqref{eq:fast-mr-moments},
\[
 \E\!\left[\left|\int_0^{nT}h(X_u)\,\mathrm dW_u\right|^4\right]
 \le C\E\!\left[\left(\int_0^{nT}h(X_u)^2\,\mathrm du\right)^2\right]
 \le C(nT)^2.
\]
Consequently,
\[
 v_n=T\bar v+O(n^{-1}),\qquad
 \|A^n-v_n\|_{L^4}=O(n^{-1/2}).
\]
In the stationary case $v_n=T\bar v$ exactly.
Jensen's inequality also gives
\[
 \E[(A^n)^{-2}]
 \le\frac1{T^3}\int_0^T\E[f(X_{nt})^{-2}]\,\mathrm dt,
\]
which is bounded uniformly in $n$.  With $r_n=n^{-1/2}$,
H\"older's inequality therefore yields
\[
 \E[|Y^n|^2]
 \le\frac{n}{v_n}\|A^n-v_n\|_{L^4}^2
                  \E[(A^n)^{-2}]^{1/2}=O(1).
\]
Proposition~\ref{prop:Lp} verifies Assumption~\ref{ass:qv}.

For the smoothing condition, take
\[
 J_t^n=\int_0^t\sqrt{f(X_s^n)}\rho(X_s^n)\,\mathrm dW_s^n,
 \qquad
 N_t^n=\int_0^t\sqrt{f(X_s^n)(1-\rho(X_s^n)^2)}
                    \,\mathrm dW_s^{\perp,n}.
\]
Conditionally on $\sigma(X_0,W_u:0\le u\le nT)$, $N^n$ is Gaussian
with independent increments and
\[
 D_T^n=\int_0^Tf(X_s^n)(1-\rho(X_s^n)^2)\,\mathrm ds
       \ge\underline c A^n.
\]
Thus Assumption~\ref{ass:smoothing} holds, and for any bounded
sequence $\{k_n\}$ and $K_n=S_0e^{k_n}$,
\begin{align}
 \widehat w^n(k_n,T)
 &=\E[A^n\mid S_T^n=K_n]+o(n^{-1/2}),
 \label{eq:fast-mr-total}\\
 \widehat\sigma^n(k_n,T)^2
 &=\frac1T\E[A^n\mid S_T^n=K_n]+o(n^{-1/2}).
 \label{eq:fast-mr-annualized}
\end{align}
The role of ergodicity is the averaging in
\eqref{eq:fast-mr-time-change}; the Poisson-equation moment bounds
supply the rate and uniform integrability required by the present
theorem. 

\section{Proofs of the main results}

\begin{proof}[Proof of Theorem~\ref{thm:main}]
We suppress $n$ from terminal random variables when this causes
no ambiguity.

\subsection{The terminal density}

By Assumption~\ref{ass:smoothing}, conditionally on
$\mathscr H^n$,
\[
 \log\frac{S_{T_n}^n}{S_0}
 =J_{T_n}^n+N_{T_n}^n
  -\frac12A^n
\]
is Gaussian with conditional variance $D_{T_n}^n$.
Its conditional density at $k$ is
\begin{equation}\label{eq:conditional-kernel}
 h^n(k)
 =\frac1{\sqrt{D_{T_n}^n}}
  \phi\!\left(
   \frac{k-J_{T_n}^n+A^n/2}
        {\sqrt{D_{T_n}^n}}
  \right).
\end{equation}
The clock lower bound gives
\begin{equation}\label{eq:kernel-bound}
 0<h^n(k)
 \le\frac1{\sqrt{2\pi\underline c A^n}}.
\end{equation}
Consequently $S_{T_n}^n$ has the continuous positive
density
\[
 s\longmapsto\frac1s
 \E\!\left[h^n\!\left(\log\frac{s}{S_0}\right)\right].
\]
The continuous version used in the theorem is
\begin{align}
 m^n(k)
 &:=\E[A^n\mid S_{T_n}^n=S_0e^k]
 \nonumber\\
 &=\frac{\E[A^n h^n(k)]}
         {\E[h^n(k)]}
 \label{eq:conditional-A}
\end{align}
The denominator is finite because $(Z^n)^{-1/2}$ is integrable,
as noted after Assumption~\ref{ass:qv}, and by
\eqref{eq:kernel-bound}.  For the numerator, Jensen's inequality gives
\[
 \E[A^n h^n(k)]
 \le \frac{\E[\sqrt{A^n}]}{\sqrt{2\pi\underline c}}
 \le \sqrt{\frac{v_n}{2\pi\underline c}}.
\]
The bounds by $(A^n)^{-1/2}$ and $(A^n)^{1/2}$ also justify
continuity of the denominator and numerator by dominated convergence.
The terminal density is strictly positive on $(0,\infty)$, so
\[
 (K_n-S_0)_+<\E[(K_n-S_{T_n}^n)_+]<K_n.
\]
Indeed, positivity makes $S^n$ a supermartingale, so
$\E[S_{T_n}^n]\le S_0$.  Positive probability on both sides of the
strike gives strict Jensen inequality:
\[
 \E[(K_n-S_{T_n}^n)_+]
 >(K_n-\E[S_{T_n}^n])_+\ge(K_n-S_0)_+.
\]
The upper bound follows from $S_{T_n}^n>0$ almost surely.
Hence the put-implied total variance exists uniquely in $(0,\infty)$.

\subsection{A scaled density limit}

We claim that
\begin{equation}\label{eq:density-limit}
 \sqrt{v_n}\E[h^n(k_n)]
 -\phi(z_n)\longrightarrow0.
\end{equation}
Since $r_n\to0$, uniform integrability of $\{Y^n\}$ and
$\E[Z^n]=1$ give
\[
 Z^n-1=r_nY^n\sqrt{Z^n}\longrightarrow0
 \quad\text{in probability}.
\]
Indeed, both $\{Y^n\}$ and $\{\sqrt{Z^n}\}$ are tight.
Hence the continuous-martingale
central limit theorem gives
\begin{equation}\label{eq:martingale-CLT}
 \frac{M_{T_n}^n}{\sqrt{v_n}}
 \Longrightarrow \mathcal N(0,1);
\end{equation}
see, for example, Jacod and Shiryaev~\cite[Chapter~VIII]{JS}.

Take a subsequence along which $z_n\to z$.  The vectors
\[
 \left(
   \frac{J_{T_n}^n}{\sqrt{v_n}},
   \frac{D_{T_n}^n}{v_n}
 \right)
\]
are tight.  Indeed, the second component is bounded in probability by
$Z^n$.  Conditional Gaussianity makes
$N_{T_n}^n/\sqrt{v_n}$ tight, and
\eqref{eq:martingale-CLT} then gives tightness of the first component.
Along a further subsequence, denote the limit by $(U,\Gamma)$.  The
conditional characteristic-function identity implied by
\eqref{eq:conditional-time-change} is
\[
 \E\!\left[e^{itM_{T_n}^n/\sqrt{v_n}}\right]
 =\E\!\left[
   e^{itJ_{T_n}^n/\sqrt{v_n}
      -t^2D_{T_n}^n/(2v_n)}
 \right].
\]
Passing to the limit gives
\begin{equation}\label{eq:normal-mixture}
 U+\sqrt\Gamma G\ \stackrel{d}{=}\ \mathcal N(0,1),
\end{equation}
where $G$ is standard normal and independent of $(U,\Gamma)$, and
$\Gamma\ge\underline c$.

Since $\{Y^n\}$ and $\{\sqrt{Z^n}\}$ are tight,
\[
 \frac{A^n-v_n}{\sqrt{v_n}}
 =r_n\sqrt{v_n}\,
   Y^n\sqrt{Z^n}\longrightarrow0
 \quad\text{in probability}.
\]
It follows from \eqref{eq:conditional-kernel} that
\[
 \sqrt{v_n}h^n(k_n)
 \Longrightarrow
 \frac1{\sqrt\Gamma}\phi\!\left(
       \frac{z-U}{\sqrt\Gamma}\right).
\]
The random variable on the right is the conditional density of
$U+\sqrt\Gamma G$ at $z$.  Since $\Gamma\ge\underline c$, its
expectation is continuous in $z$; equality in law in
\eqref{eq:normal-mixture} therefore makes it equal to $\phi(z)$ at
every $z$.  Moreover,
\[
 \sqrt{v_n}h^n(k_n)
 \le\frac1{\sqrt{2\pi\underline c}}
      (Z^n)^{-1/2}.
\]
Uniform integrability proves \eqref{eq:density-limit} by the
subsequence criterion.  Moreover,
\begin{equation}\label{eq:conditional-shift}
 \frac{m^n(k_n)-v_n}
      {v_n r_n}
 =\frac{\E[Y^n\sqrt{Z^n}
                 h^n(k_n)]}
        {\E[h^n(k_n)]}
 =O(1),
\end{equation}
because
\[
 \sqrt{v_n}\,
 \E[|Y^n|\sqrt{Z^n}
                 h^n(k_n)]
 \le\frac{\E[|Y^n|]}{\sqrt{2\pi\underline c}}.
\]

\subsection{The option-price expansion}

We extend the stopped Black--Scholes argument of
\cite[Section~4.1]{Fukasawa2026} to the present variance clocks.

Define
\[
 L_t^n=\cE(J^n)_t=\exp\{J_t^n-C_t^n/2\}.
\]
Orthogonality gives $S_t^n=S_0L_t^n\cE(N^n)_t$.
Let $p_n(s,q)=\Pbs(s,K_n,q)$, and write
$C=C^n$, $D=D^n$.  Set
\[
 b_n=(1-\underline c/2)v_n,
 \qquad
 \tau_n=\inf\{t:C_t\ge b_n\},
 \qquad
 q_t^n=v_n-C_{t\wedge\tau_n}.
\]
Then $q_t^n\ge\underline c v_n/2$.  By
\eqref{eq:BS-PDE} and It\^o's formula,
\[
 p_n(S_0L_{t\wedge\tau_n}^n,q_t^n)
\]
is a bounded local martingale, hence a martingale.  In particular,
\begin{equation}\label{eq:stopped-BS}
 \E[p_n(S_0L_{T_n\wedge\tau_n}^n,
                   q_{T_n}^n)]
 =p_n(S_0,v_n).
\end{equation}

Choose $\eta_0\in(0,1)$ such that
$(1-\underline c)(1+\eta_0)<1-\underline c/2$, and put
\[
 G_n=\{|Z^n-1|\le\eta_0\}.
\]
On $G_n$ we have
$C_{T_n}\le(1-\underline c)A^n<b_n$ and hence
$\tau_n>T_n$.  Since a put is bounded by its strike,
\eqref{eq:short-time-localization},
\eqref{eq:stopped-BS}, and conditional Gaussianity give
\begin{align}
 &\E[(K_n-S_{T_n}^n)_+]
      -p_n(S_0,v_n)
 \nonumber\\
 &\quad=
 \E\left[
 p_n(S_0L_{T_n}^n,D_{T_n})
 -p_n(S_0L_{T_n}^n,
             v_n-C_{T_n});G_n
 \right]+o\!\left(K_n r_n
                         \sqrt{v_n}\right).
 \label{eq:localized-price}
\end{align}

Both variance arguments in \eqref{eq:localized-price}, divided by
$v_n$, lie in a fixed compact subset of $(0,\infty)$.  From
\eqref{eq:BS-PDE}, after scaling $q=v_n u$, the family
\[
 F_n(s,u):=\frac{\sqrt{v_n}}{K_n}
 \partial_qp_n(s,v_n u)
 =\frac{\phi(d)}{2\sqrt u},
 \qquad
 d=\frac{x}{\sqrt u}-\frac{\sqrt{v_n u}}2,
 \quad x=\frac{\log(s/K_n)}{\sqrt{v_n}},
\]
satisfies
\[
 \partial_uF_n(s,u)
 =\frac{(d^2-1)\phi(d)}{4u^{3/2}}
  +\frac{\sqrt{v_n}\,d\phi(d)}{4u}.
\]
Consequently, on this compact set,
\[
 |F_n(s,u)|\le C,
 \qquad
 |\partial_uF_n(s,u)|
 \le C(1+\sqrt{v_n}),
\]
uniformly over $s>0$.  Taylor's formula at $D_{T_n}$ and
\[
 D_{T_n}-(v_n-C_{T_n})
 =A^n-v_n
 =v_n r_n Y^n\sqrt{Z^n}
\]
therefore yield
\begin{align}
 &p_n(S_0L_{T_n}^n,D_{T_n})
 -p_n(S_0L_{T_n}^n,
             v_n-C_{T_n})
 \nonumber\\
 &\qquad=v_n r_n
       Y^n\sqrt{Z^n}
       \partial_qp_n(S_0L_{T_n}^n,
                            D_{T_n})
       +R_n,
 \label{eq:Taylor}
\end{align}
where 
\[
 R_n=K_n\sqrt{v_n}\,\Delta_n
       \int_0^1\bigl\{F_n(s_n,u_n-\theta\Delta_n)
                         -F_n(s_n,u_n)\bigr\}\,\mathrm d\theta.
\]
and  $s_n=S_0L_{T_n}^n$,
$u_n=D_{T_n}/v_n$, and $\Delta_n=Z^n-1$.

The whole segment $u_n-\theta\Delta_n$ lies in the same compact
set.  Boundedness of $F_n$ and the mean-value estimate give,
respectively,
\[
 \left|F_n(s_n,u_n-\theta\Delta_n)-F_n(s_n,u_n)\right|
 \le \min\{2C,C(1+\sqrt{v_n})\theta|\Delta_n|\}.
\]
Consequently,
\[
 |R_n|\le C K_n\sqrt{v_n}|\Delta_n|
      \min\{1,(1+\sqrt{v_n})|\Delta_n|\}
 \quad\text{on }G_n.
\]
Using $|\Delta_n|/r_n=|Y^n|\sqrt{Z^n}$ and
$\sqrt{Z^n}\le\sqrt{1+\eta_0}$ on $G_n$ yields
\[
 \frac{|R_n|}{K_n r_n\sqrt{v_n}}
 1_{G_n}
 \le C|Y^n|
 \min\!\left\{1,
  (1+\sqrt{v_n})|Z^n-1|\right\}1_{G_n}.
\]
The factor inside the minimum converges to zero in probability because
\[
 (1+\sqrt{v_n})|Z^n-1|
 =(r_n+r_n\sqrt{v_n})
   |Y^n|\sqrt{Z^n}.
\]
Uniform integrability in Assumption~\ref{ass:qv} therefore gives
\[
 \E[|R_n|;G_n]
 =o\!\left(K_n r_n\sqrt{v_n}\right).
\]
Furthermore, \eqref{eq:BS-PDE} and
\eqref{eq:conditional-kernel} give the exact identity
\begin{equation}\label{eq:vega-kernel}
 \partial_qp_n(S_0L_{T_n}^n,
                      D_{T_n})
 =\frac{K_n}{2}h^n(k_n).
\end{equation}
The contribution of $G_n^c$ to
$\E[Y^n\sqrt{Z^n}
        h^n(k_n)]$ is $o(v_n^{-1/2})$
by uniform integrability in Assumption~\ref{ass:qv},
\eqref{eq:kernel-bound}, and
$\Q(G_n^c)\to0$.  Combining the preceding displays gives
\begin{equation}\label{eq:price-expansion}
 \E[(K_n-S_{T_n}^n)_+]
 =p_n(S_0,v_n)
  +\frac{v_n r_n K_n}{2}
       \E[Y^n\sqrt{Z^n}
                    h^n(k_n)]
  +o\!\left(K_n r_n\sqrt{v_n}\right).
\end{equation}

\subsection{Inversion}

By \eqref{eq:BS-PDE},
\begin{equation}\label{eq:reference-vega}
 \partial_qp_n(S_0,v_n)
 =\frac{K_n}{2\sqrt{v_n}}\phi(z_n).
\end{equation}
Equations \eqref{eq:conditional-A}, \eqref{eq:density-limit}, and
\eqref{eq:price-expansion} imply
\begin{equation}\label{eq:price-at-conditional-A}
 \E[(K_n-S_{T_n}^n)_+]
 =p_n(S_0,m^n(k_n))
  +o\!\left(K_n r_n\sqrt{v_n}\right).
\end{equation}
Indeed, \eqref{eq:conditional-shift} gives
$m^n(k_n)/v_n=1+O(r_n)$.  On this scale,
the derivative bound preceding \eqref{eq:Taylor} and
$r_n(1+\sqrt{v_n})\to0$ justify the first-order
Black--Scholes expansion.  Its leading term is
\[
 \partial_qp_n(S_0,v_n)
 \bigl(m^n(k_n)-v_n\bigr)
 =\frac{K_n v_n r_n}{2}
   \E[Y^n\sqrt{Z^n}h^n(k_n)]
   \frac{\phi(z_n)}
        {\sqrt{v_n}\E[h^n(k_n)]}.
\]
The last factor tends to one by \eqref{eq:density-limit}.

To invert \eqref{eq:price-at-conditional-A}, fix $\delta>0$.  Since
$m^n(k_n)/v_n=1+O(r_n)$ and $r_n\to0$, the lower endpoint below
is positive for all sufficiently large $n$.  For all
$q$ between
\[
 m^n(k_n)-\delta v_n r_n
 \quad\text{and}\quad
 m^n(k_n)+\delta v_n r_n,
\]
we have $q/v_n=1+O(r_n)$.  Moreover,
\[
 d_2(S_0,K_n,v_n u)
 =-\frac{z_n}{\sqrt u}
   +\frac{\sqrt{v_n}(1-u)}{2\sqrt u}.
\]
Thus $\partial_qp_n(S_0,q)$ is bounded below by
$cK_n/\sqrt{v_n}$ throughout this interval.  Strict
monotonicity and \eqref{eq:price-at-conditional-A} imply, for every
$\delta>0$ and all sufficiently large $n$,
\[
 \left|\widehat w^n(k_n,T_n)
       -m^n(k_n)\right|
 \le\delta v_n r_n.
\]
Consequently,
\[
 \widehat w^n(k_n,T_n)
 -m^n(k_n)=o(v_n r_n).
\]
This is \eqref{eq:main-result}.
\end{proof}

\section*{Acknowledgment of the use of generative AI}

Generative AI was used to assist with language editing and the
organization of an initial draft.  The author remains responsible for
the mathematical content.

\end{document}